\documentclass[a4paper]{article}
\usepackage[T1]{fontenc} 
\usepackage{RRA4,RRthemes}
\usepackage{hyperref}

\usepackage{amsmath,amssymb,amsthm}
\usepackage{rotating}
\usepackage{textcomp} 
\usepackage{color,graphicx,url} 

\newcommand{\Cbf}{\mathit{Cbf1}}
\newcommand{\Ash}{\mathit{Ash1}}
\newcommand{\Swi}{\mathit{Swi5}}
\newcommand{\GalF}{\mathit{Gal4}}
\newcommand{\GalE}{\mathit{Gal80}}
\newcommand{\gal}{\mathit{gal}}
\newcommand{\cl}[1]{\mathcal{#1}}
\newtheorem{prop}{Prop.}
\newcommand{\sign}{\mathit{sign}}
\newcommand{\AG}{\mathbf{AG}}
\newcommand{\AF}{\mathbf{AF}}
\newcommand{\AX}{\mathbf{AX}}
\newcommand{\EX}{\mathbf{EX}}
\newcommand{\EF}{\mathbf{EF}}

\newcommand{\low}{\mathit{low}}
\newcommand{\high}{\mathit{high}}

\RRdate{April 2010}

\RRauthor{
 Gregory Batt\footnote{INRIA Paris - Rocquencourt, Le Chesnay, France}
 \and Michel Page\footnote{INRIA Grenoble - Rh\^{o}ne-Alpes, Montbonnot, France and Universit\'{e} Pierre Mend\`{e}s France, Grenoble, France}
 \and Irene Cantone\footnote{Clinical Sciences Center, Imperial College, London, United Kingdom} 
 \and Gregor Goessler\footnote{INRIA Grenoble - Rh\^{o}ne-Alpes, Montbonnot, France}
 \and Pedro Monteiro\footnote{INRIA Grenoble - Rh\^{o}ne-Alpes, Montbonnot, France and INESC/Instituto Superior T\'ecnico, Lisbon, Portugal}
 \and Hidde de Jong\footnote{INRIA Grenoble - Rh\^{o}ne-Alpes, Montbonnot, France}
}  
\authorhead{Batt \textit{et~al}}

\RRtitle{Paramétrisation efficace de modèles qualitatifs de réseaux de régulation à l'aide de model checking symbolique}
\RRetitle{Efficient Parameter Search\\ for Qualitative Models of Regulatory Networks\\ using Symbolic Model Checking\footnote{Contact: gregory.batt@inria.fr}}
\titlehead{Efficient Parameter Search for Qualitative Models of Regulatory Networks}
\RRabstract{
Investigating the relation between the structure and behavior of complex biological networks often involves posing the following two questions: Is a hypothesized structure of a regulatory network consistent with the observed behavior? And can a proposed structure generate a desired behavior? 
Answering these questions presupposes that we are able to test the compatibility of network structure and behavior.

We cast these questions into a parameter search problem for qualitative models of regulatory networks, in particular piecewise-affine differential equation models. 
We develop a method based on symbolic model checking that avoids enumerating all possible parametrizations, and show that this method performs well on real biological problems, using the IRMA synthetic network and benchmark experimental data sets. 
We test the consistency between the IRMA network structure and the time-series data, and search for parameter modifications that would improve the robustness of the external control of the system behavior. 

GNA and the IRMA model are available at http://www-helix.inrialpes.fr/gna
}
\RRkeyword{
systems and synthetic biology,
qualitative models of gene networks,
symbolic model checking,
model validation and system design
}
\RRresume{
L'étude du lien entre structure et fonctionnement de réseaux biologiques complexes revient souvent à poser l'une des deux questions suivantes:
Est ce qu'une structure hypothétique est cohérente avec des comportements observés? Est-ce qu'un comportement désiré peut être obtenu avec une structure proposée?
Nous formulons ce problème en un problème de recherche de paramè-tres pour des modèles qualitatifs de réseaux de régulations.

Nous développons une méthode basée sur des techniques de model checking symbolique qui évite de devoir énumérer toutes les paramétrisations possibles, et nous montrons que cette méthode est bien adaptée aux problèmes biologiques réels à l'aide du réseau synthétique IRMA et des données expérimentales de référence associées.
Nous testons la cohérence entre la structure du réseau IRMA et les données de séries temporelles, et nous cherchons des modifications des paramètres qui rendent le système plus robuste à un contrôle externe par addition de galactose.

GNA et le modèle IRMA sont disponibles en ligne à l'adresse http://www-helix.inrialpes.fr/gna
}
\RRmotcle{
biologie systémique et synthétique,
modèles qualitatifs de réseaux de gènes,
model checking symbolique,
validation de modèles et conception de réseaux 
}
\RRprojets{Contraintes, Pop Art and Ibis }
\RRdomaine{5c} 
\RRthemeProj{ibis}
\RRdomaineProjBis{ibis} 

\URRocq 
\RCParis

\begin{document}
\RRNo{7284}
\makeRR   

\section{Introduction}

A central problem in the analysis of biological regulatory networks
concerns the relation between their structure and dynamics. This
problem can be narrowed down to the following two questions: (a) Is a
hypothesized structure of the network consistent with the observed
behavior? (b) Can a proposed structure generate a desired behavior?

Qualitative models of regulatory networks, such as (synchronous or
asynchronous) Boolean models and piecewise-affine differential
equation (PADE) models, have been proven useful for answering the
above questions. 
The models are coarse-grained, in the sense that they do not specify the biochemical mechanisms in
detail. However, they include the logic of gene regulation and allow
different expression levels of the genes to be distinguished. 
They are interesting in their own right, as a way to capture in a simple 
manner even complex dynamics induced by the structure of interactions,
for example steady states and transient responses to external perturbations. 
They can also be used as a first step to orient the development of more
fine-grained quantitative ODE models. Several applications of logical
and PADE models have confirmed their interest for the study of large
and complex regulatory networks \cite{gb97,gb76,HdJ2964,HdJKlamt}.

Qualitative models bring specific advantages over numerical models when studying the relation
between structure and dynamics. In order to answer questions (a) and
(b), one has to search the parameter space to check if for some
parameter values the network can be consistent with the data or a
desired control objective can be attained. 
In qualitative models the number of different parametrizations is finite 
and the number of possible values for each parameter is usually rather low.
This makes parameter search easier to handle than in quantitative models, where exhaustive search
of the continuous parameter space is in general not
feasible. Moreover, much of the available data in biology is
semi-quantitative rather than fully quantitative due to variability in the experimental conditions and biological
material, imprecise and relative measurements, low sampling density,
... Qualitative models are more concerned with qualitative trends in
the data rather than with precise quantitative values.

Nevertheless, the parametrization of qualitative models remains a
complex problem. For large models, the state and parameter spaces are
usually too large to test all combinations of parameter values using
existing techniques. This makes it difficult to answer questions (a)
and (b) for most networks of actual biological interest. The aim of
this paper is to address this search problem for PADE models by
treating it in the context of formal verification and symbolic model checking.
More specifically, we formulate the dynamic properties in temporal logic
and verify by means of a model checker if the network satisfies these properties  \cite{HdJ2003,HdJ2963}.

Our contributions are twofold. 
On the methodological side, we develop a method that in comparison with our previous work \cite{HdJ2427} makes it possible to analyze
very efficiently models with a large state space, and even to analyze incompletely parametrized models without the need for exhaustive enumeration of all parametrizations.
 This is achieved by a symbolic encoding of the model structure, the constraints on parameter values (if available), and the transition rules
describing the qualitative dynamics of the PADE models. We can thus take full advantage of symbolic model checkers for testing the consistency of the network structure
with dynamic properties expressed in temporal logics. 
The current version 8 of GNA~\cite{HdJMonteiro} has been extended with export functionalities to generate the symbolic encoding of PADE models in the NuSMV language~\cite{gr71}.
In comparison with related work \cite{gb137,HdJ2379, gb136,gbFromentin}, our method applies to incompletely instead of fully parametrized models, provides more precise results, and the  encoding is efficient without (strongly) simplifying the PADE dynamics.

On the application side, we show that the method performs well on real problems, by means of the IRMA synthetic network and benchmark experimental data sets
\cite{HdJCantone}. More precisely, we are able to find parameter values for which the network satisfies temporal-logic properties describing observed expression
profiles, both on the level of individual and averaged time-series. 
The method is selective in the sense that only a small part of the parameter space is found to be compatible with the observations. 
Analysis of these parameter values reveals that biologically-relevant constraints have been identified. 
Moreover, we make suggestions to improve the robustness of the external control of the IRMA behavior by proposing a rewiring of the network.

\section{Qualitative model of IRMA network}
\label{sec:model}

\subsection{IRMA network}
\label{sec:model,network}

IRMA is a synthetic network constructed in yeast and proposed as a
benchmark for modeling and identification approaches
\cite{HdJCantone}. The network consists of five well-characterized
genes that have been chosen so as to include different kinds of
interactions, notably transcription regulation and protein-protein
interactions.  The endogenous copies of the genes were deleted, so as
to reduce crosstalk of IRMA with the regulatory networks of
the host cell. In order to further isolate the synthetic network from
its cellular environment, the genes belong to distinct and
non-redundant pathways. Moreover, they are non-essential, which means
that they can be knocked out without affecting yeast viability.

The structure of the IRMA network is shown in
Fig.~\ref{fig:network}(a). The expression of the \textit{CBF1} gene is
under the control of the \textit{HO} promoter, which is positively
regulated by Swi5 and negatively regulated by Ash1. \textit{CBF1}
encodes the transcription factor Cbf1 that activates expression of the
\textit{GAL4} gene from the \textit{MET16} promoter. 
The \textit{GAL10} promoter is activated by Gal4, but only  in the absence of Gal80 or in the presence of galactose which releases the inhibition of Gal4 by Gal80.
The \textit{GAL10} promoter controls the expression of \textit{SWI5},
whose product not only activates the above-mentioned \textit{HO}
promoter, but also the \textit{ASH1} promoter 
controling the expression of both the \textit{GAL80} and \textit{ASH1} genes.

Notice that the network contains both positive and negative feedback
loops. Negative feedback loops are a necessary condition for the
occurrence of oscillations \cite{HdJThomas}, while the addition of
positive feedback loops is believed to increase the robustness of the
oscillations \cite{gb134}. This suggests that the network structure, for
suitable parameter values, might be able to function as a synthetic
oscillator.

\begin{figure}
\mbox{(a)} \includegraphics[width=0.80\linewidth]{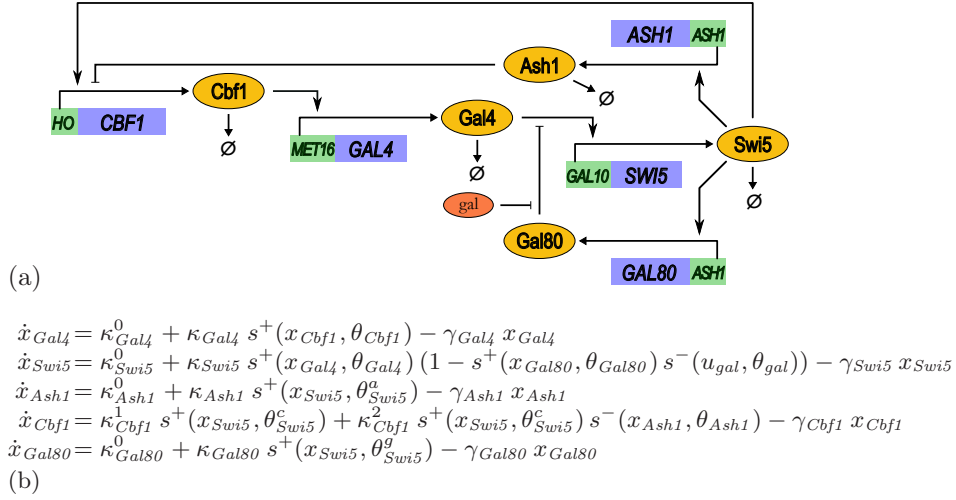}\\
\mbox{}\\
\begin{small}
$\begin{array}[b]{@{}r@{}l} 
\dot{x}_\GalF & = \kappa_\GalF^0 + \kappa_\GalF\, s^{+}(x_\Cbf,\theta_\Cbf) - \gamma_\GalF\, x_\GalF \\ 
\dot{x}_\Swi & = \kappa_\Swi^0 + \kappa_\Swi\, s^{+}(x_\GalF,\theta_\GalF) \, (1-s^{+}(x_\GalE,\theta_\GalE) \,s^{-}(u_\gal,\theta_\gal)) - \gamma_\Swi\, x_\Swi \\ 
\dot{x}_\Ash & = \kappa_\Ash^0 + \kappa_\Ash\, s^{+}(x_\Swi,\theta_\Swi^a) - \gamma_\Ash\, x_\Ash \\
\dot{x}_\Cbf & = \kappa_\Cbf^1\,s^{+}(x_\Swi,\theta_\Swi^c) + \kappa_\Cbf^2\, s^{+}(x_\Swi,\theta_\Swi^c) \,s^{-}(x_\Ash,\theta_\Ash) -\gamma_\Cbf\, x_\Cbf \\ 
\dot{x}_\GalE & = \kappa_\GalE^0 + \kappa_\GalE\, s^{+}(x_\Swi,\theta_\Swi^g) - \gamma_\GalE\, x_\GalE 
\end{array}$
\mbox{(b)}
\end{small}
\caption{Synthetic IRMA network in yeast. (a) Schematic representation of the network constructed in \cite{HdJCantone}. 
The green and blue boxes are promoter and genes, and the yellow and red ovals are proteins and metabolites. 
(b) PADE model of IRMA, with state variables $x$, protein synthesis constants $\kappa$, decay constants $\gamma$, and thresholds $\theta$. The input variable $u_\gal$ refers to the presence of galactose ($\dot{u}_\gal=0$). 
The subscripts $_\GalF$, $_\Swi$, $_\Ash$, $_\Cbf$, $_\GalE$ refer to the proteins with the same name.} 
\label{fig:network}
\end{figure}

\subsection{Measurements of IRMA dynamics}
\label{sec:model,data}

The behavior of the network has been monitored  in response to two
different perturbations \cite{HdJCantone}: shifting cells from glucose
to galactose medium (switch-on experiments), and from galactose to
glucose medium (switch-off experiments). The terms 'switch-on'
('switch-off') refer to the activation (inhibition) of \textit{SWI5} expression
during growth on galactose (glucose). For these two perturbations, the
temporal evolution of the expression of all the genes in the network
was monitored by qRT-PCR with good time resolution:
samples every 10 min (switch-off) or 20 min (switch-on) over more than 3 h.

Fig.~\ref{fig:average_data}(a) represents the expression of all genes, averaged over 5 (switch-on) or 4 (switch-off) independent experiments. 
In the switch-off experiments (galactose to glucose), the transcription of all genes is shut down. In the switch-on experiments, a seemingly oscillatory
behavior is present with Swi5 peaks at 40 and 180 min, while Swi5, Cbf1, and Ash1 are expressed at moderate to high levels \cite{HdJCantone}. 
We remark that the confidence intervals are large (not shown), which means that the data are essentially semi-quantitative.

\begin{figure}
\includegraphics[width=\linewidth]{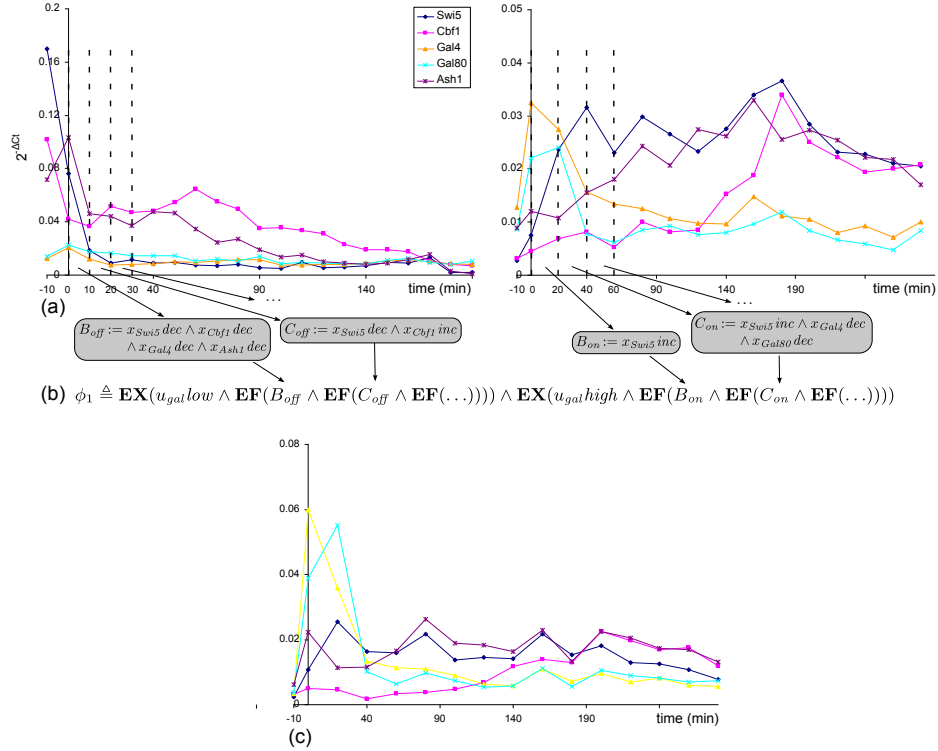}
\caption{Dynamic behavior of the IRMA network in response to medium shift perturbations. 
(a) Temporal profiles of averaged gene expression measured with qRT-PCR during switch-off (left) and switch-on (right) experiments (data from \cite{HdJCantone}). 
(b) Temporal logic encoding of the switch-off and switch-on behaviors.
The operator $\EF \, \phi$ expresses the possibility to reach a future state satisfying $\phi$, whereas the operator $\EX \, \phi$ is used to require the existence
of an initial state satisfying $\phi$. $u_{gal} \, \low$ and $u_{gal} \, \high$ denote the absence and presence of galactose, respectively. 
See \cite{HdJ2003} for more details on the temporal logic CTL. 
(c) Temporal profile of gene expression in an individual switch-on experiment showing a switch-off-like behavior. 
} 
\label{fig:average_data}
\end{figure}

The analysis of the individual time-series reveals that in some
cases the gene expression profiles are indeed similar, at least
qualitatively, whereas in other cases notable differences are observed 
(for example, the oscillatory behavior is not present in all switch-on time-series; Fig.~\ref{fig:average_data}(c)). 
In the latter case, average expression levels may be a misleading representation of the network behavior.

\subsection{PADE model of IRMA network}
\label{sec:model,PADE}

We built a qualitative model of the IRMA dynamics using PADE models of
genetic regulatory networks. The PADE models, originally introduced in
\cite{HdJ1030}, provide a coarse-grained picture of the network dynamics. 
They have the following general form:

\begin{equation}
\dot{x}_i= f_i(x) \triangleq \sum_{l\in L_i}\kappa_i^l\; b_i^l(x) - \gamma_i\; x_i, \ i \in [1,n]
\label{eq:diffeq}
\end{equation}
where $x \in \Omega \subset \Bbb{R}_{\geq 0}^n$ represents a vector of $n$ protein (or RNA) concentrations. The synthesis rate is composed of a sum of synthesis
constants $\kappa_i^l$, each modulated by a regulation function $b_i^l(x)\in \{0,1\}$. A regulation function is an algebraic expression of
step functions $s^+(x_j,\theta_j)$ or $s^-(x_j,\theta_j)$ which formalizes the regulatory logic of gene expression. $\theta_j$ is a so-called threshold for the
concentration $x_j$. The step function $s^+(x_j,\theta_j)$ evaluates to 1 if $x_j>\theta_j$, and to 0 if $x_j<\theta_j$, thus capturing the switch-like character of
gene regulation ($s^-(x_j,\theta_j)=1-s^+(x_j,\theta_j)$). The degradation of a gene product is a first-order term, with a degradation constant $\gamma_i$ that includes
contributions of growth dilution and protein degradation. The models can be easily extended to account for proteolytic regulators, but we will omit this here as IRMA
does not include such factors.

In the case of IRMA, we define five variables that correspond to the total protein concentrations of Cbf1, Gal4, Gal80, Ash1, and Swi5, as well as an input variable
denoting the concentration of galactose.
Notice that the measurements of the network dynamics concern mRNA and not protein levels. We assume that the variations
in mRNA and protein levels are the same, even though this may not always be the case. A similar approximation is made in \cite{HdJCantone}, where protein and mRNA
levels are considered to be proportional.

The PADE model of the IRMA network is shown in Fig.~\ref{fig:network}(b). 
Consider for example the equation for the protein Gal4.
$\kappa_\GalF^0$ is its basal synthesis rate, and $\kappa_\GalF^0 + \kappa_\GalF$ its maximal synthesis rate when the \textit{GAL4} activator Cbf1 is present (\textit{i.e.}, $x_\Cbf>\theta_\Cbf$). 
Swi5 is regulated in a more complex way. 
The expression of its gene is activated by Gal4, but 
only when Gal80 is absent or galactose present (which prevents Gal4 inactivation by Gal80), that is, only when not both Gal80 is present and galactose absent.
The step-function expression in Fig.~\ref{fig:network}(b) mathematically describes this condition.
We remark that for the regulation of \textit{CBF1}, we take into account that Ash1 can override the effect of Swi5, that is, inhibition dominates activation. Moreover, Swi5 is assumed to have three different thresholds, for the regulation of \textit{CBF1}, \textit{GAL80}, and \textit{ASH1}.

The PADE model is a direct translation of the IRMA network into a simple mathematical format. The model resembles the ODE model in \cite{HdJCantone}, but notably approximates
the Hill-type kinetic rate laws by step functions. It thus makes the implicit assumption that important qualitative dynamical properties of the network are intimately
connected with the network structure and the regulatory logic, independently from the details of the kinetic mechanisms and precise parameter values. Several studies
have shown this assumption to be valid in a number of model systems \cite{gb97,HdJBornholt,gb122}, although care should be exercised in deciding exactly when
modeling approximations are valid \cite{HdJPolynikis}.

To investigate for the possible existence of unknown interactions between the synthetic network and the host,
we would like to test given the PADE model above whether the network structure and the regulatory logic can account for the qualitative trends in the gene expression data observed in \cite{HdJCantone}.
Because in some experiments it has been observed that the addition of galactose does not always lead to an activation of the IRMA genes, we also search for parameter modifications that renders the network response to an addition of galactose more robust.

\section{Search of parameter space using symbolic model checking}
\label{sec:model_checking}

\subsection{Qualitative analysis of PADE models}

The advantage of PADE models is that the qualitative dynamics of high-dimens-ional systems are relatively easy to analyze, using only the total order on parameter values rather than exact numerical values \cite{HdJ2725,HdJEdwards}. 
The main difficulty lies in treating the discontinuities in the right-hand side of the differential equations, at the threshold values of the step functions. Following \cite{HdJ1899}, the use of differential inclusions based on Filippov solutions has been proposed in \cite{HdJ2725} and implemented in the computer tool GNA \cite{HdJ2427}.
Here, we recast this analysis in a form that underlies the symbolic encoding of the dynamics below.

The key to our reformulation of the qualitative analysis of the PADE dynamics is the extension of step functions $s^+$
to interval-valued functions $S^+$ , where
\begin{equation}
S^+(x_j,\theta_j)= \left\{\begin{array}[c]{l}
\mbox{}[0,0] \mbox{ if } x_j <\theta_j,\\
\mbox{}[0,1] \mbox{ if } x_j =\theta_j,\\
\mbox{}[1,1] \mbox{ if } x_j >\theta_j
\end{array}\right.
\end{equation}
That is, because the step functions are not defined at their thresholds, we conservatively assume that they can take any value between 0 and 1 (see \cite{gb97}
for a similar idea). When replacing the step functions by their extensions, the regulation functions $b_i^l(x)$ become interval-valued functions
$B_i^l\,:\,\Bbb{R}_{\geq 0}^n \rightarrow \{[0,0],[0,1],[1,1]\}$, and Eq.~(\ref{eq:diffeq}) generalizes to the following differential inclusion:

\begin{equation} \dot{x}_i \in F_i(x) \triangleq \sum_{l\in L_i}\kappa_i^l\;
B_i^l(x) - \gamma_i\; x_i, \ i \in [1,n]
\label{eq:diffincl}
\end{equation}
As shown in Section~\ref{sec:comparison}, for most models the solutions of this differential inclusion are the same as the solutions of the differential inclusions defined in \cite{HdJ2725}.

The starting-point for our qualitative analysis is the introduction of a rectangular partition $\cl{D}$ of the state space $\Omega$. This partition is induced  by the union of the two sets $\Theta_i$ and $\Lambda_i$, $i\in [1,n]$, where $\Theta_i = \{ \theta_i^j \mid j \in J_i \}$ and $\Lambda_i = \{ \sum_{l\in B} \kappa_i^l/\gamma_i
\mid B \subseteq L_i\}$. That is, the partition is a rectangular grid defined by the threshold parameters $\theta_i^j$ and the so-called focal parameters $\sum_{l\in B}
\kappa_i^l/\gamma_i$. The focal parameters are steady-state concentrations towards which the PADE system locally converges in a monotonic way \cite{HdJ1030}. 
For the variable $x_\GalF$, we have $\Theta_\GalF = \{ \theta_\GalF\}$ and $\Lambda_\GalF = \{0, \kappa_\GalF^0/\gamma_\GalF,(\kappa_\GalF^0 +
\kappa_\GalF)/\gamma_\GalF \}$. 

Interestingly, the partition has the property that in each domain $D\in \cl{D}$, the protein production rates are identical: for all
$x,y\in D$, it holds that $B_i^l(x)=B_i^l(y)\triangleq B_i^l(D)$. 
As a consequence, the derivatives of the concentration variables have a unique sign
pattern: for all $x,y\in D$, it holds that $\sign(F_i(x))=\sign(F_i(y))\subseteq \{-1,0,1\}$  \cite{HdJ2725}. 
Notice that this property is not obtained for less fine-grained partitions
used in related work \cite{gb137,HdJ2379,gb97,gb136,gb76,gbFromentin,gb138}. 
It will be seen to be critical for the search of parametrized models of IRMA that satisfy the time-series data.

The above considerations motivate a discrete abstraction, resulting in a state transition graph. 
In this graph, the states are the domains $D \in \cl{D}$,
and there is a transition from a domain $D$ to another domain $D'$, if there exists a solution of the differential inclusion in Eq. \ref{eq:diffincl} that starts in $D$ and reaches $D'$, without leaving $D \cup D'$. 
$$D\rightarrow D' \mbox{ iff }
\begin{array}[t]{l}
\exists \xi \mbox{ solution of } (\ref{eq:diffincl}), \exists \tau\in \mathbb{R}_{\neq 0}\cup \{\infty\}\\
\mbox{such that } \xi(0)\in D, \xi(\tau)\in D', \mbox{ and } \forall  t\in [0,\tau], \xi(t)\in D\cup D'
\end{array}$$
The state transition graph defines the qualitative dynamics of the system, in the sense that the states give a qualitative description of the state of the system (derivative sign patterns, threshold and focal parameters bounding the domain), while paths in the state transition graph describe how this state evolves over time (changes in derivative patterns, changes in bounds of domain) \cite{HdJ2725}.

We reformulate here the transition rules using the interval extensions of the regulation functions.
The existence of a transition depends on the sign of $F$ at the boundary between the two domains.
To capture this notion, we introduce an interval-valued function $F_i\,:\,\cl{D}\times\cl{D} \rightarrow 2^{\Bbb{R}}$, where $F_i(D,D')=\sum_{l\in L_i}\kappa_i^l\; B_i^l(D) - \gamma_i\, D'_i$, for $D,D' \in \cl{D}$. $F_i(D,D')$ represents the flow in $D$ infinitely close to $D'$.
In order to evaluate $F_i(D,D')$, we use interval arithmetic \cite{HdJ1385}.
For instance, in a domain in which $x_\Swi>\theta_\Swi^c$ and $x_\Ash=\theta_\Ash$, we have $S^{+}(x_\Swi,\theta_\Swi^c)=[1,1]$ and $S^{-}(x_\Ash,\theta_\Ash)=[0,1]$, so that the differential inclusion for $x_\Cbf$ becomes $[\kappa_\Cbf^1,\; \kappa_\Cbf^1+\kappa_\Cbf^2] - \gamma_\Cbf \cdot x_\Cbf$. 

We distinguish three types of transitions, depending on whether the transition goes from a domain $D$ to itself ($D=D'$, internal transition), from a domain
$D$ to another, higher-dimensional domain $D'$ ($D \subseteq \partial D'$, dimension-increasing transition), or from a domain $D$ to another, lower-dimensional domain $D'$ ($D' \subseteq \partial D$, dimension-decreasing transition), 
where $\partial D$ denotes the boundary of $D$ in its supporting hyperplane.
For dimension-increasing transition, we obtain the following rule:

\begin{prop}[Dimension-increasing transition]\rm Let $D,D'\in \cl{D}$ and $D \subseteq \partial D'$. $D \rightarrow D'$ is a dimension-increasing transition iff
\begin{enumerate}
\item $\forall i\in [1,n]$, such that $D_i$ and $D'_i$ coincide with a value in $\Theta_i \cup \Lambda_i$, it holds that $0\in F_i(D',D)$, and
\item $\forall i\in [1,n]$, such that $D_i\neq D'_i$, it holds that $\exists \alpha>0$ such that $\alpha \in F_i(D',D) \, (D'_i - D_i)$
\end{enumerate}
\label{prop:diminc}
\end{prop}
The first condition guarantees that solutions can remain in domains located in threshold and focal planes, while the second condition expresses that the
direction of the flow in the domains ($F_i(D',D)$) is consistent with their relative position ($D'_i - D_i$). 
The rules for other types of transitions and their proofs can be found in Section~\ref{sec:tr}.

It can be shown that exact parameter values are not needed for the analysis of the qualitative dynamics of a PADE model: it is sufficient to know the
ordering of the threshold and focal parameters \cite{HdJ2725}. This comes from the fact that the sign of $F_i$, and hence the transitions and the state transition graph, are invariant for regions of the parameter space defined by a particular total order on $\Theta_i \cup \Lambda_i$ \cite{HdJ2725}. We call such a total order a parametrization of the PADE model.

\subsection{Search of parameter space: a model-checking approach}

Verifying the compatibility of the network structure with an observed or desired behavioral property (Section~\ref{sec:model,data}) can be achieved by comparing the state transition graph with qualitative trends in the data. For large graphs like that obtained for
IRMA (which has about 50000 states) this becomes quickly impossible to do by hand. This has motivated the use of model-checking tools (\textit{e.g.},
\cite{gb137,HdJ2427,HdJ2379,HdJ2963}). For PADE models, each state in the graph is described by atomic propositions whose truth-value are preserved under the discrete abstraction, such as the
above-mentioned derivative sign patterns. 
The atomic propositions are used to formulate observed or desired properties in a temporal-logic formula $\phi$ and model
checkers test if the state transition graph $T$ satisfies the formula ($T\models \phi$). 

Because the number of possible parametrizations and the size of state transition graphs rapidly grow with the number of genes, the naive approach consisting in enumerating all parametrizations of a PADE model, and for each of these generating the state transition graph and testing whether $T\models \phi$, is only feasible for the simplest networks.
We therefore propose an alternative approach, based on the symbolic encoding of the above search problem, without explicitly generating the possible parametrizations of the PADE models and the corresponding state transition graphs. 
This enables one to exploit the capability of symbolic model checkers to efficiently manipulate implicit descriptions of the state and parameter space.

\subsection{Symbolic encoding of PADE model and dynamics}
\label{sec:model_checking,symbolic_encoding}

In this section we summarize the main features of the encoding . 
We particularly focus on the discretization of the state space, which connects the symbolic encoding to the mathematical analysis of PADE models, and the use of the discretization for the computation of $F_i(D',D)$ in Prop.~\ref{prop:diminc}, 
which is essential for state transition computations.

The symbolic encoding is based on a discretization of the state space implied by the partition $\cl{D}$. 
We call $\cl{C}$ a discretization function that maps $D\in \cl{D}$ to a set of unique integer coordinates, and $\cl{C}(D)= \cl{C}(D_1) \times \ldots \times \cl{C}(D_n)$. 
Let $m_i$ be the number of non-zero parameters in $\Theta_i \cup \Lambda_i$, $i\in [1,n]$.
Then $\cl{C}(D_i) \in \{ 0, 1, \ldots, 2m_i +1\}$, and more specifically, $\cl{C}(D_i) \in \{ 0, 2, \ldots, 2m_i \}$ if $D_i$ coincides with a threshold or
focal plane, and $\cl{C}(D_i) \in \{ 1, 3, \ldots, 2m_i+1 \}$ otherwise. 
More generally, $\cl{C}(S) = \{ \cl{C}(D) \mid D\subseteq S \}$, for any set of domains $S$.
Obviously, $\cl{C}$ can also be used for the discretization of parameter values. 
For example, in the case of the variable $x_\GalF$, we have one threshold and three focal parameters. 
Now, let $D$ be a domain in the state space and $D_\GalF$ its component in the $x_\GalF$-dimension. Given the following total order on the threshold
and focal parameters, $0<\kappa_\GalF^0/\gamma_\GalF<\theta_\GalF<(\kappa_\GalF^0 + \kappa_\GalF)/\gamma_\GalF$, we find $\cl{C}(0)=0$ (by definition),
$\cl{C}(\kappa_\GalF^0/\gamma_\GalF)=2$, $\cl{C}(\theta_\GalF)=4$, and $\cl{C}(\kappa_\GalF^0 + \kappa_\GalF)/\gamma_\GalF)=6$.

The above discretization motivates the introduction of symbolic variables $\hat{D}_i$, $\hat{D}'_i$, $\hat{\theta}_i^j$, $\hat{\lambda}_i^j$ that encode $\cl{C}(D_i)$,
$\cl{C}(D'_i)$, $\cl{C}(\theta_i^j)$, $\cl{C}(\lambda_i^j)$, respectively. The different conditions in Prop.~\ref{prop:diminc} can be expressed in terms of this
encoding. For instance, the sign of $D'_i-D_i$ simply becomes $\hat{D}'_i-\hat{D}_i$. The translation is less evident for the encoding of $F_i(D',D)$, the sign of
which needs to be computed in the transition rules. Multiplication by $1/\gamma_i$ does not change the sign, but gives the more convenient expression

\begin{equation}
F_i(D,D')/\gamma_i=\sum_{l\in L_i} (\kappa_i^l/\gamma_i)\; B_i^l(D) - D'_i \label{eq:diffinclsimpl}
\end{equation}
Recall that the first term in the righthand side is simply an interval whose upper and lower bound are focal parameters, determined by the regulation functions
$B_i^l(D)$. 
By redefining the step functions in terms of the symbolic variables:
\begin{equation}
S^+(D_j,\theta_j)= \left\{\begin{array}[c]{ll}
\mbox{} [0,0] & \mbox{ iff } \hat{D}_j < \hat{\theta}_j\\
\mbox{} [0,1] & \mbox{ iff } \hat{D}_j = \hat{\theta}_j\\
\mbox{} [1,1] & \mbox{ iff } \hat{D}_j > \hat{\theta}_j
\end{array}\right.
\end{equation}
each $B_i^l(D)$ can be simply computed by means of interval arithmetic. 
Evaluating the expression $\sum_{l\in L_i} (\kappa_i^l/\gamma_i)\; B_i^l(D)$ leads to an interval with focal
parameters as bounds, and which can therefore be represented by $\hat{\lambda}_i^j$.
From this interval we subtract $\hat{D}'_i$ to symbolically define $F_i(D,D')/\gamma_i$. 
The sign of the latter expression allows one to check the conditions of Prop.~\ref{prop:diminc}, and thus to
derive the transitions in the state transition graph. 
The specification of transitions in a symbolic way is the main stumble block for the efficient encoding of the PADE dynamics, especially when $D$ is located on a threshold plane. 
In our previous work~\cite{HdJ2725}, the computation of transitions required the enumeration of an exponential number of domains surrounding $D$ \cite{gb137}. 
The interval-based formulation proposed here avoids this inefficient approach and allows $F_i(D,D')/\gamma_i$ to be computed in one stroke.

The implementation in a model checker like NuSMV \cite{gr71} is straightforward with the above encoding. In particular, we apply invariant constraints on the symbolic
variables to exclude all valuations of $\hat{D}_i$, $\hat{D}'_i$, $\hat{\theta}_i^j$, $\hat{\lambda}_i^j$ that do not correspond to a valid transition from $D$ to $D'$
for a given parameterization. We apply three types of invariants. The first one constrains parameters to remain constant. The second one constrains $D$ and $D'$ to
be neighbors in the state space (\textit{e.g.}, $D\subseteq \partial D'$ for dimension-increasing transitions). 
The final invariant constrains the relative position of $D$ and $D'$ and the parameter order as stated in the transitions conditions. 
For comparison with experimental data, we also need to know 
the variations of concentrations of gene products in each state. 
Formally, it is defined as the derivative sign pattern, and simply corresponds to the sign of $F_i(D,D')$ as computed above.

The initial states of our symbolic structure correspond to each possible parametrization and transitions towards all states $D$. A CTL property $\phi$ holds for a
symbolic structure if all initial states satisfy $\phi$. Therefore, by testing whether $\neg \phi$ holds, we verify the absence of a parametrization satisfying $\phi$.
A counterexample to $\neg \phi$ thus returns a valid parametrization. 
The current version 8 of GNA~\cite{HdJMonteiro} has been extended with export functionalities to generate the symbolic encoding of PADE models in the NuSMV language.

\section{Validation: consistency of IRMA network with experimental data}

Are the observations of the IRMA dynamics consistent with the network structure? At first sight this question may seem incongruous as one expects this to be the
case by definition (each genetic construct was tested before integration in the yeast cell). However, in practice it is far from trivial, even if the design and
construction have been carried out with great care, to avoid interactions between the synthetic network and the host.

\subsection{Temporal-logic encoding of observations}

To test the consistency between our PADE model of the IRMA network and the experimental data, we express that for each condition, switch-on and switch-off,
there must exist an initial state of the system and a path starting from this state along which the gene expression changes correspond to the observed time-series
data. 
For example, for the switch-off time-series we encode that there exists an initial state where in absence of galactose the expression of
\textit{SWI5}, \textit{CBF1}, \textit{GAL4} and \textit{ASH1} decreases (in the interval $[0,10]$ min), and from which a state can be reached where the expression of
\textit{SWI5} decreases and the expression of \textit{CBF1} increases (in the interval $[10,20]$ min), \textit{etc}. 
The generation of this property, called $\phi_1$, from the experimental data leads to the temporal-logic formula shown in Fig.~\ref{fig:average_data}(b).
The property is automatically generated from the experimental data using a Matlab script. 

To disregard small fluctuations due to biological and experimental noise, we considered that changes of magnitude less that $5\cdot 10^{-3}$ units are not
significant. This smooths out, for example, Gal4 expression levels in switch-off conditions after 40 min. 
In \cite{HdJCantone} it was demonstrated by glucose-to-glucose shift experiments that the mere resuspension of cells into fresh medium has a network-independent effect:  the expression of \textit{GAL80} and \textit{GAL4} is strongly increased in the first 10 min after resuspension. Therefore, we did not incorporate in our specification the very first measurements (in the interval $[-10,0]$) made just before shifting cells to a new medium.

The data presented in \cite{HdJCantone} for switch-on and switch-off conditions are the average of 5 and 4 individual experiments, respectively. As noticed in
Section~\ref{sec:model,data}, expression profiles obtained in similar conditions may differ significantly. 
In the case of such heterogeneous behavior, properties capturing the average gene expression profile may be misleading. Consequently, asking for consistency between our model and the result of each individual experiment might be more appropriate. This leads us to define a second property $\phi_2$ similar to $\phi_1$ but requiring the existence of 9 paths in the graph, one
for each of the observed behaviors in the 5 switch-on and 4 switch-off experiments. 
Although the information we extract from the experimental data is purely qualitative, only concerning trends in gene expression levels, the accumulation of these simple observations leads to fairly complex constraints. 
Property $\phi_2$ involves nearly 160 constraints on derivative signs.

\subsection{Testing consistency of network with observations}

We use our symbolic encoding of the PADE dynamics and verify the existence of a valid parameter ordering. 
We do this by testing the negation of $\phi_1$ or $\phi_2$, such that a
negative answer from the model checker proves the existence of at least one valid parametrization, as explained in Section~\ref{sec:model_checking,symbolic_encoding}.
Moreover, the counterexample returned provides one such parametrization. 
By means of this approach, we can prove the existence of a parametrization satisfying the averaged time-series data ($\phi_1$). 
The result was obtained in 49\,s on a laptop (PC, 2.2\,Ghz, 1 core, 2\,Gb RAM). 
The state space contains nearly 50000 discrete states and the parameter space is discretized into nearly 5000 different parameter orderings. 
The counterexample of $\neg \phi_1$, obtained in 100 s, provides a valid parametrization (Table~\ref{tab:result_table}).

When analyzing the corresponding parametrization, the thresholds are mostly higher than the focal parameter for basal expression and lower than the focal parameter for upregulated expression, \textit{e.g.}, $\kappa_\Ash^0/\gamma_\Ash < \theta_\Ash < (\kappa_\Ash^0+\kappa_\Ash)/\gamma_\Ash$. This is not surprising as the focal parameters correspond to the lowest and highest possible expression levels. The threshold at which Ash1 controls \textit{CBF1} expression is expected to lie between the two extremes. 
The only exception in the parameters found by the model checker is Gal80, for which it holds $(\kappa_\GalE^0+\kappa_\GalE)/\gamma_\GalE < \theta_\GalE$. According to this constraint, Gal80 plays no role in the system, since it cannot exceed the threshold concentration above which it inhibits Swi5. 
This is interesting because it suggests that the switch-off behavior may occur even without any inhibition by Gal80, and consequently, in a galactose-independent manner.

The dynamic properties of the PADE model can be analyzed in more detail by means of GNA. 
This shows the existence of an asymptotically stable steady state
corresponding to switch-off conditions, with low Swi5, Gal4, Cbf1, Ash1, and Gal80 concentrations. 
In addition, GNA finds strongly connected components (SCCs) consistent with the observed damped oscillations observed in galactose media. 
However, the attractors co-exist irrespectively of the presence or absence of galactose, revealing that galactose does not necessarily drive the system to a single attractor for this particular parametrization.

We also tested whether the above parametrization is consistent with time-series data from the individual experiments.  
In 3\,s the model checker shows that it does not satisfy the more constraining property $\phi_2$. 
However, we do find another parametrization for which $\phi_2$ holds. In this case, all thresholds are situated between
the basal and upregulated focal parameters (237\,s, including counterexample generation).

\begin{sidewaystable}
\begin{scriptsize}
\begin{tabular}{|c|c|c|c|c|}
\hline
 & \multicolumn{2}{c|}{Symbolic state space and symbolic parameter space} & \multicolumn{2}{c|}{Symbolic state space and explicit parameter space}\\
\hline
Property & Existence of & Parametrization$^*$ & Number of & Parametrization$^*$ \\
& parametrization & & parametrizations &\\
\hline
\raisebox{4mm}{} $\phi_1$: averaged & Yes & $\frac{\kappa_\Swi^0}{\gamma_\Swi} < \theta^g_\Swi < \theta^c_\Swi <\theta^a_\Swi < \frac{\kappa_\Swi^0+\kappa_\Swi}{\gamma_\Swi}$ & 12 & $\frac{\kappa_\Swi^0}{\gamma_\Swi} < \theta^c_\Swi< \theta^a_\Swi < \frac{\kappa_\Swi^0+\kappa_\Swi}{\gamma_\Swi} ~\wedge$\\ 
 time-series & (49\,s) & $\wedge ~\frac{\kappa_\GalE^0}{\gamma_\GalE}< \frac{\kappa_\GalE^0+\kappa_\GalE}{\gamma_\GalE} < \theta_\GalE$ & (925\,s) & $(
\begin{array}[t]{r@{~}l}
&\theta_\GalE < \frac{\kappa_\GalE^0}{\gamma_\GalE} \wedge \frac{\kappa_\Swi^0}{\gamma_\Swi} < \theta^g_\Swi <\frac{\kappa_\Swi^0+\kappa_\Swi}{\gamma_\Swi}\\
\vee& \frac{\kappa_\GalE^0}{\gamma_\GalE}< \theta_\GalE < \frac{\kappa_\GalE^0+\kappa_\GalE}{\gamma_\GalE}  \wedge \frac{\kappa_\Swi^0}{\gamma_\Swi} < \theta^g_\Swi\\
\vee& \frac{\kappa_\GalE^0+\kappa_\GalE}{\gamma_\GalE} < \theta_\GalE )\raisebox{-2mm}{}
\end{array}$ \\
\hline
\raisebox{4mm}{}$\phi_2$: individual & Yes & $\frac{\kappa_\Swi^0}{\gamma_\Swi} <\theta^c_\Swi < \theta^a_\Swi< \theta^g_\Swi < \frac{\kappa_\Swi^0+\kappa_\Swi}{\gamma_\Swi}$ & 4 & $\frac{\kappa_\Swi^0}{\gamma_\Swi} <\theta^c_\Swi < (\theta^a_\Swi, \theta^g_\Swi) < \frac{\kappa_\Swi^0+\kappa_\Swi}{\gamma_\Swi}$\\
\raisebox{-2mm}{}time-series & (131\,s) & $\wedge ~\frac{\kappa_\GalE^0}{\gamma_\GalE}< \theta_\GalE< \frac{\kappa_\GalE^0+\kappa_\GalE}{\gamma_\GalE}$ & (2021\,s) & $\wedge~(\frac{\kappa_\GalE^0}{\gamma_\GalE}, \theta_\GalE) < \frac{\kappa_\GalE^0+\kappa_\GalE}{\gamma_\GalE}$\\
\hline
\raisebox{4mm}{}$\phi_3$: single &  Yes & $\theta^c_\Swi  < \frac{\kappa_\Swi^0}{\gamma_\Swi} < \theta^g_\Swi< \theta^a_\Swi <\frac{\kappa_\Swi^0+\kappa_\Swi}{\gamma_\Swi}$ & 7 & $\theta^c_\Swi  < \frac{\kappa_\Swi^0}{\gamma_\Swi} < \theta^a_\Swi <\frac{\kappa_\Swi^0+\kappa_\Swi}{\gamma_\Swi}$\\
attractor & (126\,s) &$\wedge~ \theta_\GalE <\frac{\kappa_\GalE^0}{\gamma_\GalE}<\frac{\kappa_\GalE^0+\kappa_\GalE}{\gamma_\GalE}$ & (1300\,s)& $\wedge~\theta_\GalE <\frac{\kappa_\GalE^0+\kappa_\GalE}{\gamma_\GalE}$\\
\raisebox{-2mm}{}&&&& $\wedge~(\theta^g_\Swi <\frac{\kappa_\Swi^0}{\gamma_\Swi}\vee\theta_\GalE <\frac{\kappa_\GalE^0}{\gamma_\GalE}~)$\\
\hline
\end{tabular}\\
$^*$All parametrizations additionally include 
$\kappa^1_\Cbf/\gamma_\Cbf < \theta_\Cbf <(\kappa^1_\Cbf+\kappa^2_\Cbf)/\gamma_\Cbf 
~\wedge~ \kappa_\GalF^0/\gamma_\GalF < \theta_\GalF <(\kappa_\GalF^0+\kappa_\GalF)/\gamma_\GalF 
~\wedge~ \kappa_\Ash^0/\gamma_\Ash < \theta_\Ash < (\kappa_\Ash^0+\kappa_\Ash)/\gamma_\Ash$.
\end{scriptsize}
\caption{
Summary of parametrizations found by checking the consistency of the IRMA structure with the observed and desired behaviors, expressed as temporal-logic
properties $\phi_1$, $\phi_2$, and $\phi_3$. 
The table shows the parametrization returned when testing the truth-value of the property on the symbolically encoded PADE model and
dynamics (left) and summarizes all parametrizations satisfying the properties (right).
} \label{tab:result_table}

\end{sidewaystable}

\subsection{Detailed analysis of valid parameter set}

As stated above, our consistency tests only confirm that a parametrization exists for which the structure of the network is consistent with the observed behavior.
However, it does not say if this is trivially the case (when most parametrizations are) or if the properties are selective (when most parametrization are not). To
investigate this we exhaustively generated all possible parametrizations, and tested for each of them property $\phi_1$ (averaged time-series) and $\phi_2$ (individual time-series). Although the total number of parameter orderings (4860) is fairly large, the exhaustive analysis is still manageable for networks of this size.

Out of the 4860 completely parametrized PADE models, we found that only a surprisingly small subset is consistent with the observations. 
For the averaged time-series, only 12 parametrizations are consistent, while for the individual time-series this subset is further reduced to 4
parametrizations (Table~\ref{tab:result_table}). The properties extracted from the data are thus seen to be quite selective.

These results indicate that to be consistent with the experimental
data, the activation threshold of \textit{CBF1} by Swi5, $\theta^c_\Swi$, must
be smaller than the activation thresholds of \textit{ASH1} and \textit{GAl80} by Swi5,
$\theta^a_\Swi$ and $\theta^g_\Swi$.  
Interestingly, this result is corroborated by independent experimental studies. 
Fitting of experimental data on promoter activities to Hill functions showed that the activation threshold for the \textit{ASH1} promoter, controlling \textit{ASH1} and \textit{GAL80} expression, is nearly twice as high as the one for the \textit{HO} promoter controlling \textit{CBF1} expression (Table S1 of \cite{HdJCantone}).

A second finding is that the dynamics of the system is consistent with the experimental data even if $\theta_\GalE < \kappa_\GalE^0/\gamma_\GalE$, that is when \textit{GAL80} is constitutively expressed above its inhibition threshold.
This indicates that an effective regulation of \textit{GAL80} expression by Swi5 is of little importance for the functioning of the network.
And indeed, it was found that \textit{GAL80} is not much responsive to changes in Swi5 availability: Cantone \textit{et al} observed that a 6-fold
increase of \textit{SWI5} expression leads to only a negligible (1.08-fold) increase in \textit{GAL80} expression levels (Fig. 4A in \cite{HdJCantone}).

\section{Re-engineering: improving external control by galactose}

As stated above, it has been experimentally observed that the system response to an addition of galactose is not always identical.
In one experiment at least, the addition of galactose does not significantly changes the system's behavior: a switch-off like response is observed in switch on conditions (Fig.~\ref{fig:average_data}(c)).
To obtain a more robust external control of the system, we would like to ensure that the addition of galactose drives the system out of the low-Swi5 state.

\subsection{Temporal-logic specification of design objective}

We start by specifying that two attractors can be reached, one in switch-off conditions, and one in switch-on conditions.
In switch-off conditions, the Swi5 concentration must eventually remain low, that is, equal to its basal expression level $\kappa_\Swi^0/\gamma_\Swi$. This is expressed in CTL as $\AF{} \, \AG{} \, x_\Swi\, \textit{low}$.
In switch-on conditions, an oscillatory behavior in the concentration of Swi5 is expected. It can be formulated by means of the formula $\AG{} \, \AF{}
\, (x_\Swi \, \textit{inc} \, \wedge \, \AF{} x_\Swi \, \textit{dec})$, requiring that an increase in $x_{Swi5}$ is observed infinitely often and is necessarily followed by a decrease in $x_\Swi$. 
In addition to these two basic requirements, we impose that in presence of galactose, the Swi5 concentration cannot indefinitely stay low: 
$u_\gal\, \high \rightarrow \AF{} \, \neg x_\Swi\, \textit{low}$.
We prefix these specifications so as to express the possibility ($\EX{}$) to reach the appropriate attractor from some initial state, and the necessity ($\AX{}$) to leave the switch-off steady state for all initial states in switch-on conditions. This gives rise to the following property:

$\begin{array}{rl}
\phi_3= &\EX \, (u_\gal\, \high \, \wedge \AG \, \AF \, (x_\Swi \, \textit{inc} \, \wedge \, \AF x_\Swi \, \textit{dec})) \\
&\wedge \, \EX \, (u_\gal\, \low \wedge \AF \, \AG \, x_\Swi\, \textit{low})\\
&\wedge \, \AX{} \, (u_\gal\, \high \rightarrow   \AF{} \, \neg x_\Swi\, \textit{low})
\end{array}$

\subsection{Parametrizations consistent with design objective}

Using symbolic model checking, we test the feasibility of $\phi_3$. In about 2\,min, we find that parametrizations exist for which the system presents the desired
behavior (Table~\ref{tab:result_table}). 
Using GNA, we can analyze the proposed parameter ordering. 
In the presence of galactose, several SCCs are found, with two terminal SCCs attracting the major part of the state space that includes notably the switch-off state: from a switch-off initial state, oscillations necessarily happen. 
In the absence of galactose, a unique stable steady-state where all genes are off is attracting the entire state space. 
Indeed, although SCCs are present, they are non-terminal and one can show that the switch-off steady state is eventually always reached.

As explained above one of the time-series in the switch-on conditions contradicts our specification.
It is consequently not surprising that none of the parametrizations consistent with the experimental data satisfies our design requirements, suggesting that changes are needed.
We therefore tried to find other parametrizations, consistent with $\phi_3$. 
Our method indeed finds an order on the threshold and focal parameters satisfying the property (proven in 126 s), while the enumeration of all 4860 parametrizations shows that only 7 are valid (1300\,s; Table~\ref{tab:result_table}).

A first surprising feature is that $\theta_\Swi^c<\kappa_\Swi^0/\gamma_\Swi$: Swi5 must always activate \textit{CBF1}. 
Stated differently, this constraint simply suggests to remove the regulation of \textit{CBF1} by Swi5.
This can be explained by a qualitative analysis of the system dynamics. 
In the presence of galactose, we expect oscillations for Swi5. However, the presence of Swi5 is required for
the expression of \textit{CBF1} since the \textit{HO} promoter functions like an AND gate: \textit{HO} is on if and only if Swi5 is present \emph{and} Ash1 is absent.
So, if Swi5 is not permanently present, Cbf1 and then Gal4 might diseappear, causing the system to converge to the switch-off state. 

A second surprising feature is that the regulation of \textit{GAL80} by Swi5 should not be effective.
Indeed $\theta^g_\Swi  < \kappa_\Swi^0/\gamma_\Swi$ or $\theta_\GalE < \kappa_\GalE^0/\gamma_\GalE$ means that either the \textit{GAL80} promoter is always activated, or that the Gal80 concentration is always sufficient to repress \textit{SWI5}.
As above, this suggests to remove an interaction, namely the regulation of \textit{GAL80} by Swi5. 
Interestingly, the demand for increased external control of the system leads us to a simplified design in which two out of the three feedback loops are removed.

\section{Discussion}

We proposed a method for efficient search of the parameter space of qualitative models of genetic regulatory networks. 
This allows us to test whether a hypothesized structure of the network is consistent with the observed behavior, or whether a proposed structure can generate a desired behavior.

On the methodological side, the main novelty is that we develop a symbolic encoding of the dynamics of PADE models, enabling the use of highly efficient model-checking tools for analyzing incompletely parametrized models. 
The symbolic encoding avoids the explicit generation of the state space, and the enumeration of possible parametrizations.
Although developed for PADE models, the main ideas underlying the approach carry over to logical models~\cite{HdJThomas}.

On the biological side, we show the practical relevance of our approach by means of an application to the IRMA network. The parameter constraints we obtained are precise, have a clear biological interpretation, and are consistent with independent experimental observations. Even when considering complex dynamical properties, the search of the parameter space takes at most a few minutes. 
Our results seem to confirm the intended separation of IRMA from the host network, and suggest that to obtain a more robust response to the addition of galactose, an effective rewiring of the network would be needed.

In comparison with traditional quantitative modeling approaches, the results we obtain are quite general, since they do not depend on a specific representation of the molecular details of the interactions and on specific parameter values. 
Moreover, the analysis is exhaustive in the sense that the entire parameter space is scanned. 
These two features are particularly interesting for negative results, such as showing that a given design is not likely to present a desired  behavior.
In contrast, quantitative ODE models like those developed in \cite{HdJCantone} do not predict a range of possible behaviors but rather single out one likely behavior with quantitative traits.
Qualitative and quantitative approaches provide complementary information on system dynamics.

In comparison with other analysis and verification methods developed for similar modeling formalisms \cite{gb137,HdJ2379,gb136,gbFromentin}, our approach is original in two respects. 
First, it applies to incompletely parametrized models and can handle any dynamical property of the network expressible in the temporal logic supported by the model checker. Second, we reason at a finer abstraction level, in that we take into account dynamics on the thresholds and work with a partition of the state space preserving derivative sign patterns. The latter feature is particularly well-suited for the comparison of model predictions with time-series data in IRMA.

An interesting direction for further research would be to consider even more general problems in which not only parameters but also regulation functions are incompletely specified. This would make a connection with work on the reverse engineering of Boolean models (\textit{e.g.}, \cite{gbREVEAL,gb131,gr436}.

\section*{Acknowledgment}
We would like to thank Delphine Ropers, Maria Pia Cosma, and Diego di Bernardo for helpful discussions and contributions to this work. 
We acknowledge financial support by the European Commission COBIOS FP6-2005-NEST-PATH-COM/043379. 

\appendix
\section*{Appendix}

\section{Transition rules}
\label{sec:tr}

\begin{prop}[Internal transition]\rm Let $D,D'\in \cl{D}$ and $D=D'$. $D \rightarrow D'$ is a internal transition iff
\begin{itemize}
\item[] $\forall i\in [1,n]$, such that $D_i$ coincides with a value in $\Theta_i \cup \Lambda_i$, it holds that $0\in F_i(D,D)$
\end{itemize}
\label{prop:int}
\end{prop}

\begin{proof}[Sketch of the proof]\mbox{}\\
(Necessity): If for some $i\in[1,n]$, $D_i$ coincides with a value in $\Theta_i \cup \Lambda_i$ and $0\notin F_i(D,D)$, then 
because $\forall x\in D$, $F_i(x)=F_i(D,D)$, any solution $\xi$ of (\ref{eq:diffincl}) starting in $D$ satisfies $\dot{\xi}_i(0)\neq 0$, and consequently instantaneously leaves $D$.\\
(Sufficiency): We only need to show that for some $x\in D$, there exists a solution $\xi$ of (\ref{eq:diffincl}) that remains in $D$ for some (possibly small) time interval $[0,\tau]$. Let $x_0$ be any point in $D$. For all dimensions $i$ where $D_i$ coincides with a value in $\Theta_i \cup \Lambda_i$, choose $\xi_i(t)=x_{0_i}$ for $t\geq 0$. For all other dimensions $i$, choose any solution of the differential inclusion $\dot{x}_i \in F_i(x)$. Then for a sufficiently small   $\tau>0$, $\xi(t)=(\xi_1(t),\ldots\xi_n(t))$ remains in $D$ and is a solution of (\ref{eq:diffincl}) on $[0,\tau]$.
\end{proof}

\begin{prop}[Dimension-increasing transition]\rm Let $D,D'\in \cl{D}$ and $D \subseteq \partial D'$. $D \rightarrow D'$ is a dimension-increasing transition iff
\begin{enumerate}
\item $\forall i\in [1,n]$, such that $D_i$ and $D'_i$ coincide with a value in $\Theta_i \cup \Lambda_i$, it holds that $0\in F_i(D',D)$, and
\item $\forall i\in [1,n]$, such that $D_i\neq D'_i$, it holds that $\exists \alpha>0$ such that $\alpha \in F_i(D',D) \, (D'_i - D_i)$
\end{enumerate}
\label{prop:diminc_app}
\end{prop}

\begin{proof}[Sketch of the proof]\mbox{}\\
(Necessity): Condition 1 expresses that if $D_i$ and $D'_i$ are singletons, then any solution $\xi$ remaining in $D'$ should satisfy $\dot{\xi}_i(t)=0$. The proof is made as in Proposition~\ref{prop:int}. 
For condition 2, assume that for some $i$ such that $D_i\neq D'_i$,  all $\alpha \in F_i(D',D) \, (D'_i - D_i)$ are non-positive. 
Moreover, assume without loss of generality that $D'_i - D_i>0$.
Then denoting $\underline{\lambda}_i$ and $\overline{\lambda}_i$ the bounds of the interval $\sum_{l\in L_i} (\kappa_i^l/\gamma_i)\; B_i^l(D')$, we have for any $x\in D$, $F_i(D',D)/\gamma_i=  [\underline{\lambda}_i, \overline{\lambda}_i]-x_i\leq 0$.
We deduce that  for all $x'\in D'$, $F_i(x')/\gamma_i=[\underline{\lambda}_i, \overline{\lambda}_i]-x'_i<0$. Consequently, given the relative positions of $D$ and $D'$, no solution can enter $D'$ from $D$.

(Sufficiency): We show that there exists a solution $\xi$ of (\ref{eq:diffincl}) that starts in some $x_0\in D$ ($\xi(0)=x_0$) and enters and remains in $D'$ for some (possibly small) time interval ($\xi(t)\in D', t\in ]0,\tau]$). Let $x_0$ be any point in $D$. 
For all dimensions $i\in [1,n]$ where $D_i$ and $D'_i$ are singletons, choose $\xi_i(t)=x_{0_i}$ for $t\geq 0$. 
For all dimensions $i\in [1,n]$ where $D'_i - D_i>0$ (the case $D'_i - D_i<0$ being symmetrical), $D_i$ is a singleton, and for any $x\in D$, $\max F_i(D',D)/\gamma_i= \overline{\lambda}_i-x_i>0$ implies that $\max F_i(x')= \gamma_i (\overline{\lambda}_i- x'_i)>0$ for all $x'$ in a (possibly small) neighborhood of $x$ in $D'$. 
Then choose for $\xi_i$ the solution of the differential equation $\dot{x}_i = \gamma_i (\overline{\lambda}_i- x_i)$ with $\xi_i(0)=x_{0_i}$. 
For all other dimensions, choose any solution of the differential inclusion $\dot{x}_i \in F_i(x)$, with $\xi_i(0)=x_{0_i}$.
Then for a sufficiently small $\tau>0$, $\xi(t)=(\xi_1(t),\ldots\xi_n(t))$ starts in $D$, remains in $D'$ on $]0,\tau]$ and is a solution of (\ref{eq:diffincl}) on $[0,\tau]$.
\end{proof}

\begin{prop}[Dimension-decreasing transition]\rm Let $D,D'\in \cl{D}$ and $D'\subseteq \partial D$. $D \rightarrow D'$ is a dimension-decreasing transition iff
\begin{itemize}
\item[A)] 
  \begin{enumerate}
    \item $\forall i \in [1,n]$, such that $D_i$ and $D'_i$ coincide with a value in $\Theta_i \cup \Lambda_i$, it holds that $0\in F_i(D,D')$, and
    \item $\forall i\in [1,n]$, such that $D_i\neq D'_i$, it holds that $\exists \alpha>0$ such that $\alpha \in F_i(D,D') \, (D'_i - D_i)$
  \end{enumerate}
\item[B)] or $\forall i\in [1,n]$, $0\in F_i(D,D')$
\end{itemize}
\end{prop}

\begin{proof}[Sketch of the proof]\mbox{}\\
(Necessity): 
Assume that A1 does not hold. Then B does not hold either.
This means that for some $i \in [1,n]$ such that $D_i$ and $D'_i$ coincide with a value in $\Theta_i \cup \Lambda_i$, it holds that $0\notin F_i(D,D')$. 
Then it can be shown as in the proof of Prop.~\ref{prop:int} that no solution can remain in $D$.
Now, assume that neither A2 nor B holds. As in the proof of Prop.~\ref{prop:diminc_app} and using the same notations, we can show that not A2 implies that if $D'_i-D_i>0$, then $\overline{\lambda}_i- x'_i\leq 0$. If $\overline{\lambda}_i- x'_i< 0$, there cannot be transitions from $D$ to $D'$. 
If $\overline{\lambda}_i- x'_i=0$, one can show that the solutions of the differential equation $\dot{x}_i=\max F_i(x)$ reach $\overline{\lambda}_i= x'_i$ after an infinite amount of time. Then, obviously no solution of (\ref{eq:diffincl}) can reach $D_i'=\{x'_i\}$ in lesser time. 
But then, the asymptotic convergence towards some point $x'\in D'$ implies that for all $i\in [1,n]$, $0\in \lim_{x\to x'} F_i(x)$, and hence $0\in F_i(D,D')$. Indeed, if for some $i\in[1,n]$ and $\epsilon>0$, $F_i(D,D')>\epsilon$ (or equivalently if $F_i(D,D')<\epsilon$), any solution would leave in finite time any neighborhood in $D$ of any point $x'\in D'$.  

(Sufficiency): 
Assume that conditions A1 and A2 hold. Then one can construct as in the proof of Prop~\ref{prop:diminc_app} a solution $\xi$ that starts and remains in $D$ some time interval $[0,\tau[$ ($\xi(t)\in D, t\in [0,\tau[$) and enters in $D'$ at time $\tau$ ($\xi(\tau)\in D'$). 
Alternatively, assume that condition B holds.
Then, $\forall i\in [1,n]$, $0\in F_i(D,D')$ implies that for some $x^*\in D'$, $0\in \sum_{l\in L_i} (\kappa_i^l/\gamma_i)\; B_i^l(D)-x^*_i$, $i\in[1,n]$.
Let $x_0$ be any point in $D$ and $\xi_i(t)$ be the solution of $\dot{x}_i= \gamma_i (x^*_i - x_i)$ on $[0,\infty[$ with $\xi_i(0)={x_0}_i$, $i\in [1,n]$.
One can check that $\xi=(\xi_1,\ldots,\xi_n)$ is a solution of (\ref{eq:diffincl}) such that $\forall t\geq 0$, $\xi(t)\in D$, and $lim_{t\rightarrow \infty} \xi(t)= x^*\in D'$. 
\end{proof}

\section{Comparison with previous definition of dynamics}
\label{sec:comparison}
 
In \cite{HdJ2725}, we introduced a different definition of the dynamics in regions of step function discontinuities -that is, threshold hyperplanes-
also based on differential inclusions.
The goal of this section is to show that in most cases the differential inclusions, and hence the set of solutions, are the same for both definitions.

The definition proposed in \cite{HdJ2725} makes use of the notions of regular and singular domains. 
Intuitively speaking, singular domains are located on threshold or focal planes, contrary to regular domains.
Moreover, one defines for any singular domain $D$, $R(D)$ as the set of regular domains surrounding $D$.
We refer to \cite{HdJ2725} for the precise definition of these notions. 

For simplicity of notations, we set $R(D)=D$ for a regular domain $D$. 
Also, for having more compact notations, we introduce $g_i(D,x)$ as the value of $f$ in $D$ infinitely close to $x$: $g_i(D,x)=\lim_{y\to x, y\in D} f_i(y)$. Naturally, $g_i(D,x)$ is well defined only if $x$ is in $\overline{D}$, the closure of $D$, and $g_i(D,x)=f_i(x)$ if $x\in D$.
Then in \cite{HdJ2725}, the dynamics is defined as 
\begin{equation} 
\label{eq:di_Filippov}
\dot{x}_i\in G_i(x)\triangleq [\min_{D'\in R(D)} g_i(D',x), \max_{D'\in R(D)} g_i(D',x)]
\end{equation} 
The definition we propose in this paper is more conservative than the one given in \cite{HdJ2725}.
\begin{prop}
$$G_i(x)\subseteq F_i(x), \qquad  x\in \Omega$$
\label{prop:conservativeness}
\end{prop}

\begin{proof}[Sketch of the proof]\mbox{}\\
As was done for $g_i(D,x)$, we define $s^+(D,x_j,\theta_j)$ and $b_i^l(D,x)$ as $\lim_{y\to x, y\in D} s^+(y_j,\theta_j)$ and $\lim_{y\to x, y\in D} b_i^l(y)$, respectively.
For any $x\in D$, $D\subseteq \Omega$, one can easily show that
$$ \forall D'\in R(D),\; s^+(D',x_j,\theta_j)\in S^+(x_j,\theta_j)$$
If $D_j\neq\{\theta_j\}$, then this obviously holds, as $S^+(x_j,\theta_j)=[0,0]$ (or $[1,1]$) and for all $D'\in R(D)$, $s^+(D',x_j,\theta_j)=0$ (or $1$).
If $D_j=\{\theta_j\}$, then $S^+(x_j,\theta_j)=[0,1]$, and necessarily, $s^+(D',x_j,\theta_j)\in S^+(x_j,\theta_j)$ for all $D'\in R(D)$.
Then, from interval arithmetics, we have
$$ \forall D'\in R(D),\; b_i^l(D',x)\in B_i^l(x)$$
and finally
$$ \forall D'\in R(D),\; g_i(D',x) \in F_i(x)$$
From the definition of $G_i(x)$, we conclude that $G_i(x)\subseteq F_i(x)$.
\end{proof}

In most cases, \emph{the inclusion of Prop~\ref{prop:conservativeness} is an equality}.
In fact strict inclusion can arise only in two cases.
The first one occurs when a protein has a dual role (activator and inhibitor) with a same activity threshold on a single promoter. For example, this is the case if some $b_i^l$ term equals $s^+(x_j,\theta_j)\cdot(1-s^+(x_j,\theta_j))$. Indeed, for $x_j=\theta_j$, it holds that $B_i^l(x)=[0,1]$ whereas $\max_{D'\in R(D)} b_i^l(D',x)=0$.  
The second case occurs when a protein has a dual role with a same activity threshold on two different promoters of a gene. For example, this is the case if $f_i(x)$  equals $\kappa_i^1 s^+(x_j,\theta_j) + \kappa_i^2 (1-s^+(x_j,\theta_j))-\gamma_i x_i$. Indeed, for $x_j=\theta_j$ it holds that $F_i(x)=[0,\kappa_i^1+\kappa_i^2] - \gamma_i x_i$, whereas $G_i(x)=[\kappa_i^1, \kappa_i^2] - \gamma_i x_i$, assuming $\kappa_i^1< \kappa_i^2$.
These rare cases appear in none of the PADE models developed so far and distributed with GNA, and might be considered as modeling problems.

\bibliographystyle{plain} 
\bibliography{RR-7284}

\newpage
\tableofcontents

\end{document}